\documentclass[runningheads]{llncs}

\usepackage{times}
\usepackage{amssymb}
\usepackage{amsmath}
\usepackage{ntheorem}
\usepackage{graphics}
\usepackage{graphicx}
\usepackage{epsfig}
\usepackage{textcomp}
\usepackage{xspace}
\usepackage{paralist} 
\usepackage[usenames,dvipsnames,svgnames,table]{xcolor}
\usepackage[pdfborder={0.5 0.5 0.5}]{hyperref}
\usepackage[figure]{hypcap}
\usepackage{subfigure}
\usepackage{cleveref}
\usepackage{compress}
\usepackage{fullpage}

\renewenvironment{proof}
{{\em Proof:}}{\hspace*{\fill}$\Box$\par\vspace{2mm}}

\def\comment#1{}%
\def\withcomments{%
  \newcounter{mycommentcounter}%
   \def\comment##1{\refstepcounter{mycommentcounter}%
    \ifhmode%
     \unskip%
     {\dimen1=\baselineskip \divide\dimen1 by 2 %
       \raise\dimen1\llap{\tiny
    {\textcolor{red}{\textbf{-\themycommentcounter-}}}}}\fi%
     \marginpar[{\renewcommand{\baselinestretch}{0.8}%
       \hspace*{-2em}\begin{minipage}{1.1\marginparwidth}\footnotesize%
[\themycommentcounter]:%
\raggedright ##1\end{minipage}}]{\renewcommand{\baselinestretch}{0.8}%
       \begin{minipage}{1.1\marginparwidth}\footnotesize%
[\textcolor{red}{\themycommentcounter}]: \raggedright%
##1\end{minipage}}}%
  }

    \withcomments

\newcommand{\remove}[1]{}

\newcommand{\lr}[1]{\langle #1 \rangle}
\newcommand{\mlr}[1]{\ensuremath{\lr{#1}}\xspace}

\newcommand{\mmorph}[1]{$\mlr{#1}$\xspace}
\newcommand{\poly}[1]{\lr{#1}}
\newcommand{\mpoly}[1]{\ensuremath{\poly{#1}}}

\newcommand{\qcv}{quasi-contractible\xspace}

\newcommand{\conv}{{\sc Convexifier}\xspace}

\Crefname{lemma}{Lemma}{Lemmata}
\Crefname{figure}{Fig.}{Figs.}
\Crefname{section}{Section}{Sections}
\Crefname{cl}{Claim}{Claims}

\let\savedFn=\footnote
\renewcommand*{\footnote}[1]{%
   \textcolor{red}{\savedFn{#1}}}

\titlerunning{Morphing Planar Graph Drawings Efficiently}
\authorrunning{Angelini et al.}
\begin{document}


\title{Morphing Planar Graph Drawings Efficiently\thanks{Part of the research
was conducted in the framework of ESF project 10-EuroGIGA-OP-003 GraDR ``Graph
Drawings and Representations''.}}

\author{Patrizio Angelini\inst{1}, Fabrizio Frati\inst{2}, Maurizio
Patrignani\inst{1}, Vincenzo Roselli\inst{1}
%
  }

\institute{Engineering Department, Roma Tre University, Italy\\
\email{\{angelini,patrignani,roselli\}@dia.uniroma3.it} \and School of
Information Technologies, The University of Sydney,
Australia\\\email{brillo@it.usyd.edu.au}}
\date{}

\maketitle
\begin{abstract}
A morph between two straight-line planar drawings\remove{with the same outer
face} of the same graph is a continuous transformation from the first to the
second drawing such that planarity is preserved at all times. Each step of the
morph moves each vertex at constant speed along a straight line. Although the
existence of a morph between any two drawings was established several
decades ago, only recently it has been proved that a polynomial number of steps
suffices to morph any two planar straight-line drawings. Namely, at SODA
$2013$, Alamdari \emph{et al.}~\cite{aacdfl-mpgdpns-13-c} proved that any two
planar straight-line drawings of a planar graph can be morphed in $O(n^4)$
steps, while $O(n^2)$ steps suffice if we restrict to maximal planar graphs.

In this paper, we improve upon such results, by showing an algorithm to morph
any two planar straight-line drawings of a planar graph in $O(n^2)$ steps;
further, we show that a morphing with $O(n)$ steps exists between any two planar
straight-line drawings of a series-parallel graph.
\end{abstract}

\section{Introduction}\label{se:introduction}

A \emph{planar morph} between two planar drawings of the same plane graph is a
continuous transformation from the first drawing to the second one such that
planarity is preserved at all times. The problem of deciding whether a planar
morph exists for any two drawings of any graph dates back to $1944$, when
Cairns~\cite{c-dprc-44} proved that any two straight-line drawings of a
maximal planar graph can be morphed one into the other while maintaining
planarity. In 1981, Gr\"unbaum and Shephard~\cite{gs-tgopg-81} introduced the
concept of \emph{linear morph}, that is a continuous transformation in which
each vertex moves at uniform speed along a straight-line trajectory. With this
further requirement, however, planarity cannot always be maintained for any pair
of drawings. Hence, the problem has been subsequently studied in terms of the
existence of a sequence of linear morphs, also called \emph{morphing steps},
transforming a drawing into another while maintaining planarity. The first
result in this direction is the one of Thomassen~\cite{t-dpg-83}, who proved
that a sequence of morphing steps always exists between any two straight-line
drawings of the same plane graph. Further, if the two input drawings are convex,
this property is maintained throughout the morph, as well. However, the number
of morphing steps used by the algorithm of Thomassen might be exponential in
the number of vertices.

Recently, the problem of computing planar morphs gained increasing
research attention. The case in which edges are not required to be
straight-line segments has been addressed in~\cite{lp-mpgdbe-08}, while morphs
between orthogonal graph drawings preserving planarity and orthogonality have been
explored in~\cite{lps-mopgd-2006}. Morphs preserving
more general edge directions have been considered in~\cite{bls-mpgwped-2005}. Also, the
problem of ``topological morphing'', in which the planar embedding is
allowed to change, has been addressed in~\cite{acdp-tmpg-08}.

In a paper appeared at SODA 2013, Alamdari \emph{et
al.}~\cite{aacdfl-mpgdpns-13-c} tackled again the original setting in which
edges are straight-line segments and linear morphing steps are required.
Alamdari \emph{et al.} presented the first morphing algorithms with a polynomial
number of steps in this setting. Namely, they presented an algorithm to morph
straight-line planar drawings of maximal plane graphs with $O(n^2)$ steps and of
general plane graphs with $O(n^4)$ steps, where $n$ is the number of vertices of
the graph.

In this paper we improve upon the result of Alamdari \emph{et
al.}~\cite{aacdfl-mpgdpns-13-c}, providing a more efficient algorithm to morph
general plane graphs. Namely, our algorithms uses $O(n^2)$ linear morphing
steps. Further, we provide a morphing algorithm with a linear number of steps
for a non-trivial class of planar graphs, namely series-parallel graphs.
These two main results are summarized in the following theorems.

\begin{theorem}\label{th:sp-morphing}
Let $\Gamma_a$ and $\Gamma_b$ be two drawings of the same plane series-parallel
graph $G$. There exists a morph \mmorph{\Gamma_a,\dots,\Gamma_b} with
$O(n)$ steps transforming $\Gamma_a$ into $\Gamma_b$ .
\end{theorem}

\begin{theorem}\label{th:general_case}
Let $\Gamma_s$ and $\Gamma_t$ be two drawings of the same plane graph $G$.
There exists a morph \mmorph{\Gamma_s,\dots,\Gamma_t} with $O(n^2)$
steps transforming $\Gamma_s$ into $\Gamma_t$ .
\end{theorem}

The rest of the paper is organized as follows. Section~\ref{se:preliminaries}
contains preliminaries and basic terminology. Section~\ref{se:algorithm}
describes an algorithm to morph series-parallel graphs. Section~\ref{se:general_case}
describes an algorithm to morph plane graphs. Section~\ref{se:geometry}
provides
geometric details for the morphs described in Sections~\ref{se:algorithm}
and~\ref{se:general_case}. Finally, Section~\ref{se:conclusions} contains
conclusions and open problems.

\section{Preliminaries}\label{se:preliminaries}

\textbf{Planar graphs and drawings.} A \emph{straight-line planar drawing}
$\Gamma$ (in the following simply
\emph{drawing}) of a graph $G(V,E)$ maps vertices in $V$ to distinct points of
the plane and edges in $E$ to non-intersecting open straight-line segments
between their end-vertices. Given a vertex $v$ of a graph $G$, we denote by
$\deg(v)$ the \emph{degree} of $v$ in $G$, that is, the number of vertices
adjacent to $v$.
A planar drawing $\Gamma$ partitions the plane into connected regions
called~\emph{faces}. The unbounded face is the~\emph{external face}. Also,
$\Gamma$ determines a clockwise order of the edges incident to each
vertex. Two planar drawings are \emph{equivalent} if they determine the same
clockwise ordering of the incident edges around each vertex and if they have the
same external face. A
\emph{planar embedding} is an equivalence class of planar drawings\remove{
and, hence, is described by the clockwise order of the edges incident
to each vertex and by
the choice of the external face}.
A \emph{plane} graph is a planar graph with a given planar embedding.

\textbf{Series-parallel graphs and their decomposition.} A \emph{two-terminal
series-parallel graph} $G$ with source $s$ and target $t$
can be recursively defined as follows: (i) An edge joining two vertices $s$ and
$t$ is a two-terminal series-parallel graph. Let $G'$ and $G''$ be two
two-terminal series-parallel graphs with sources $s'$ and $s''$, and targets
$t'$ and $t''$, respectively: (ii) The \emph{series composition} of $G'$ and
$G''$ obtained by identifying $s''$ with $t'$ is a two-terminal series-parallel
graph with source $s'$ and target $t''$; and (iii) the \emph{parallel
composition} of $G'$ and $G''$ obtained by identifying $s'$ with $s''$ and $t'$
with $t''$ is a two-terminal series-parallel graph with source $s'$ and target
$t'$.

A \emph{biconnected series-parallel graph} is defined as either a single edge or
a two-terminal series-parallel graph with the addition of an edge, called
\emph{root edge}, joining $s$ and $t$. In the following we deal with biconnected
series-parallel graphs not containing multiple edges.

A \emph{series-parallel graph} is a connected graph whose
biconnected components are biconnected series-parallel graphs.

A biconnected series-parallel graph $G$ with root edge $e$ is naturally
associated with an ordered binary tree $T^b_e$ rooted at $e$, called
\emph{decomposition binary tree}. Each node of $T^b_e$, with the exception of
the one associated to $e$, corresponds to a two-terminal series-parallel graph.
Nodes of $T^b_e$ are of three types, S-nodes, P-nodes, and Q-nodes. Each Q-node
represents a single edge. Each S-node represents the series composition of the
two-terminal series-parallel graphs associated with its left and
right subtrees. Finally, each P-node represents the parallel
composition of the
two-terminal series-parallel graphs associated with its left and
right subtrees.

Observe that, a graph $G$ may admit more than one binary decomposition
tree.
Also, since all internal nodes of $T^b_e$ have degree three, if $T^b_e$ is
rerooted at any other Q-node, corresponding to an edge $e' \neq e$, the obtained
ordered binary tree $T^b_{e'}$ defines a new set of compositions yielding the
same graph $G$ with root edge $e'$.

Let $G$ be an embedded biconnected series-parallel graph and let $e$ be an edge
incident to the external face of $G$. Let $T^b_e$ be one of its binary
decomposition
trees rooted at $e$. In order to have a unique \emph{decomposition tree} $T_e$
of $G$ rooted at $e$, we merge together all adjacent P-nodes and all adjacent
S-nodes of $T^b_e$. The order of the children of an S-node of $T_e$ reflects the
order of the leaves of the subtree of $T^b_e$ induced by the merged S-nodes.
Observe that, for each P-node $\mu$ of $T_e$, the embedding of $G$ induces a
circular order on the two-terminal series-parallel graphs corresponding to the
children of $\mu$. We order the children of $\mu$ according to such an
ordering.

\textbf{Morphs and Pseudo-Morphs.} A \emph{(linear)
morphing step} \mmorph{\Gamma_1,\Gamma_2}, also referred to as \emph{linear
morph}, of two straight-line planar drawings $\Gamma_1$ and
$\Gamma_2$ of a plane graph $G$ is a continuous transformation of $\Gamma_1$
into $\Gamma_2$ such that all the vertices simultaneously
start moving from their positions in $\Gamma_1$ and, moving along a
straight-line trajectory, simultaneously stop at their positions in $\Gamma_2$
so that no crossing occurs between any two edges during the
transformation.
A \emph{morph} \mmorph{\Gamma_s,\dots,\Gamma_t} of two straight-line planar
drawings $\Gamma_s$ into $\Gamma_t$ of a plane graph $G$ is a finite sequence of
morphing steps that transforms $\Gamma_s$ into $\Gamma_t$.

Let $u$ and $w$ be two vertices of $G$ such that edge $(u,w)$ belongs to $G$
and let $\Gamma$ be a straight-line planar drawing of $G$.
The \emph{contraction} of $u$ onto $w$ results in
\begin{inparaenum}[$(i)$]
\item a graph $G'= G / (u,w)$ not containing $u$ and such that each edge
$(u,x)$ of $G$ is replaced by an edge $(w,x)$ in $G'$, and
\item a straight-line drawing $\Gamma'$ of $G'$ such that each vertex different
from $v$ is mapped to the same point as in $\Gamma$.
\end{inparaenum}
In the rest of the paper, the contraction of an edge $(u,w)$ will be only
applied if the obtained drawing $\Gamma'$ is planar. The \emph{uncontraction} of
$u$ from $w$ in $\Gamma'$ yields a straight-line planar drawing $\Gamma''$ of
$G$.
A morph in which contractions are performed, possibly together with other
morphing steps, is a \emph{pseudo-morph}.

\textbf{Kernel of a vertex.} Let $v$ be a vertex of $G$ and let $G'$ be the
graph obtained by removing $v$ and its incident edges from $G$. Let $\Gamma'$ be
a planar straight-line drawing of $G'$. The \emph{kernel} of $v$ in $\Gamma'$ is
the set $P$ of points such that straight-line segments can be drawn in $\Gamma'$
connecting each point $p \in P$ to each neighbor of $v$ in $G$ without
intersecting any edge in $\Gamma'$.

\section{Morphing Series-Parallel Graph Drawings in $\mathbf{O(n)}$
Steps}\label{se:algorithm}

The aim of this section is to prove the following theorem:

\begin{theorem}\label{th:sp-pseudo-morphing}
Let $\Gamma_a$ and $\Gamma_b$ be two drawings of the same plane series-parallel
graph $G$. There exists a pseudo-morph \mmorph{\Gamma_a,\dots,\Gamma_b} with
$O(n)$ steps transforming $\Gamma_a$ into $\Gamma_b$ .
\end{theorem}

We will show in \Cref{sse:biconnected-sp} an algorithm that, given two drawings
of the same biconnected plane series-parallel graph $G$, computes a pseudo-morph
transforming one drawing into the other. Then, in
\Cref{sse:simply-connected-sp} we extend this approach to simply-connected
series-parallel graphs, thus proving Theorem~\ref{th:sp-pseudo-morphing}.

%
%
\subsection{Biconnected Series-Parallel Graphs}\label{sse:biconnected-sp}

In this section, we show an algorithm to construct a pseudo-morph transforming
one drawing of a biconnected plane series-parallel graph into another.

Our approach consists of morphing any drawing $\Gamma$ of a biconnected plane
series-parallel graph $G$ into a ``canonical drawing'' $\Gamma^*$ of $G$ in a
linear number of steps. As a consequence, any two drawings $\Gamma_1$ and
$\Gamma_2$ of $G$ can be transformed one into the other in a linear number of
steps, by morphing $\Gamma_1$ to $\Gamma^*$ and $\Gamma^*$ to $\Gamma_2$.

A \emph{canonical drawing} $\Gamma^*$ of a biconnected plane series-parallel
graph $G$ is defined  as follows. The decomposition tree $T_e$ of $G$ is
traversed top-down and a suitable geometric region of the plane is assigned to
each node $\mu$ of $T_e$; such a region will contain the drawing of the
series-parallel graph associated with $\mu$. The regions assigned to the nodes
of $T_e$ are similar to those used in~\cite{acdfp-mdgtr-12,adkmrsw-mdgfe-13} to
construct monotone drawings. Namely, we define three types of regions: Left
boomerangs, right boomerangs, and diamonds. A \emph{left boomerang} is a
quadrilateral with vertices $N, E, S$, and $W$ such that $E$ is inside triangle
$\triangle(N, S, W)$, where $|\overline{NE}|=|\overline{SE}|$ and
$|\overline{NW}|=|\overline{SW}|$ (see
Fig.~\ref{fig:boomerang-a}). A \emph{right boomerang} is defined
symmetrically, with $E$ playing the role of $W$, and vice versa (see
Fig.~\ref{fig:boomerang-b}). A \emph{diamond} is a convex
quadrilateral with vertices $N, E, S,$ and $W$, where
$|\overline{NW}|=|\overline{NE}|=|\overline{SW}|=|\overline{SE}|$. Observe that
a diamond contains a left boomerang $N_l, E_l, S_l, W_l$ and a right boomerang
$N_r, E_r, S_r, W_r$, where $S=S_l=S_r$,
$N=N_l=N_r$, $W=W_l$, and $E=E_r$ (see Fig.~\ref{fig:boomerang-c}).

\begin{figure}[htb]
\def\sf {.38}
\subfigure[]{\includegraphics[scale=\sf]{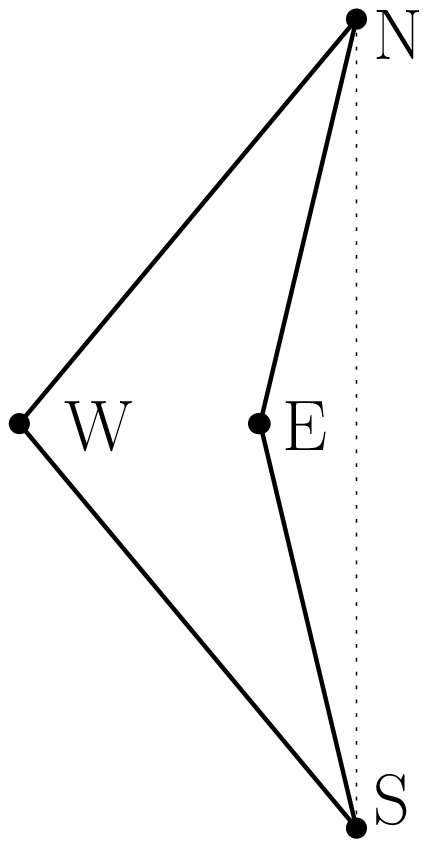}\label{fig:boomerang-a}
}\hfill
\subfigure[]{\includegraphics[scale=\sf]{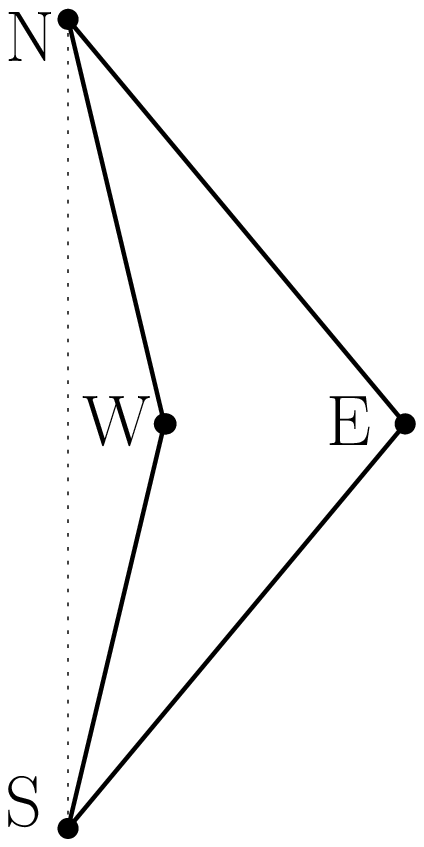}\label{fig:boomerang-b}
}\hfill
\subfigure[]{\includegraphics[scale=\sf]{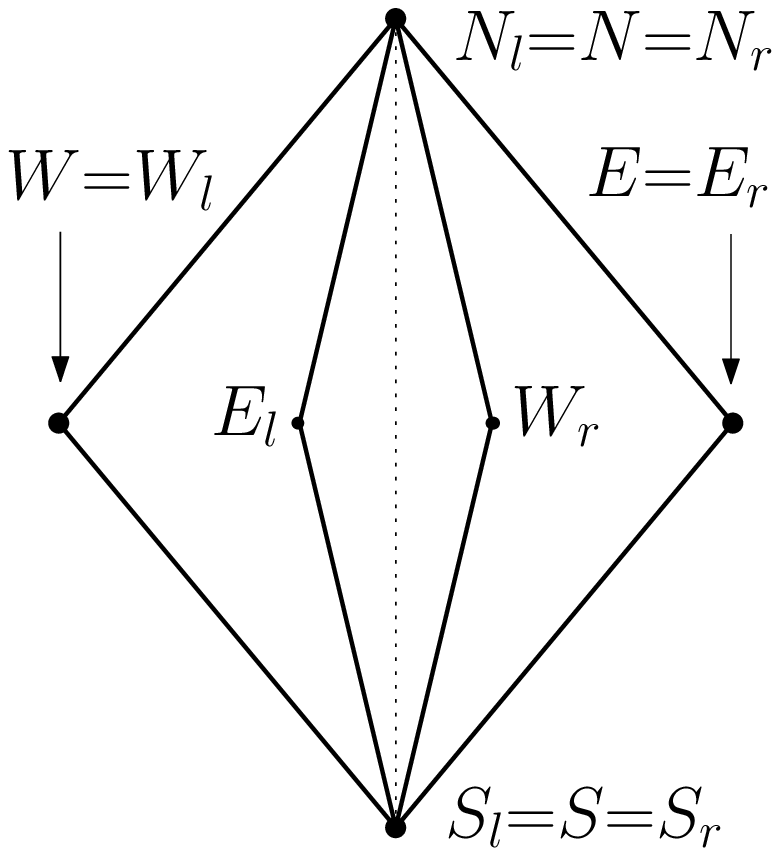}\label{fig:boomerang-c}}\hfill
\subfigure[]{\includegraphics[scale=\sf]{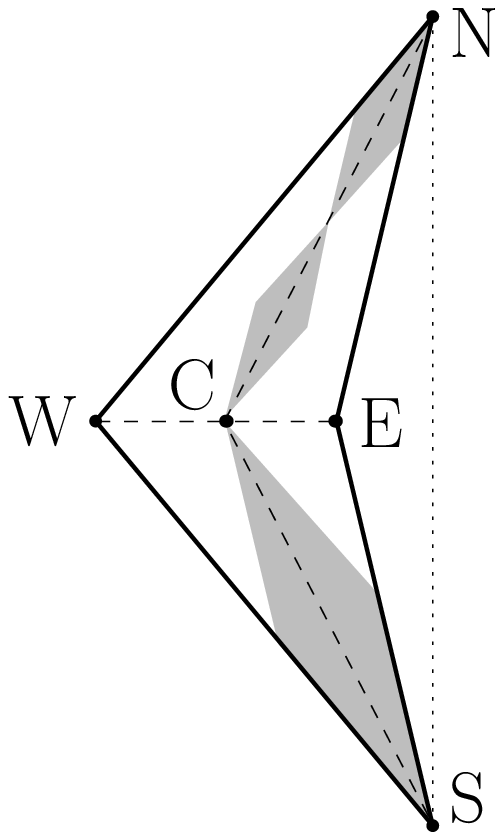}\label{fig:boomerang-d}}\hfill
\subfigure[]{\includegraphics[scale=\sf]{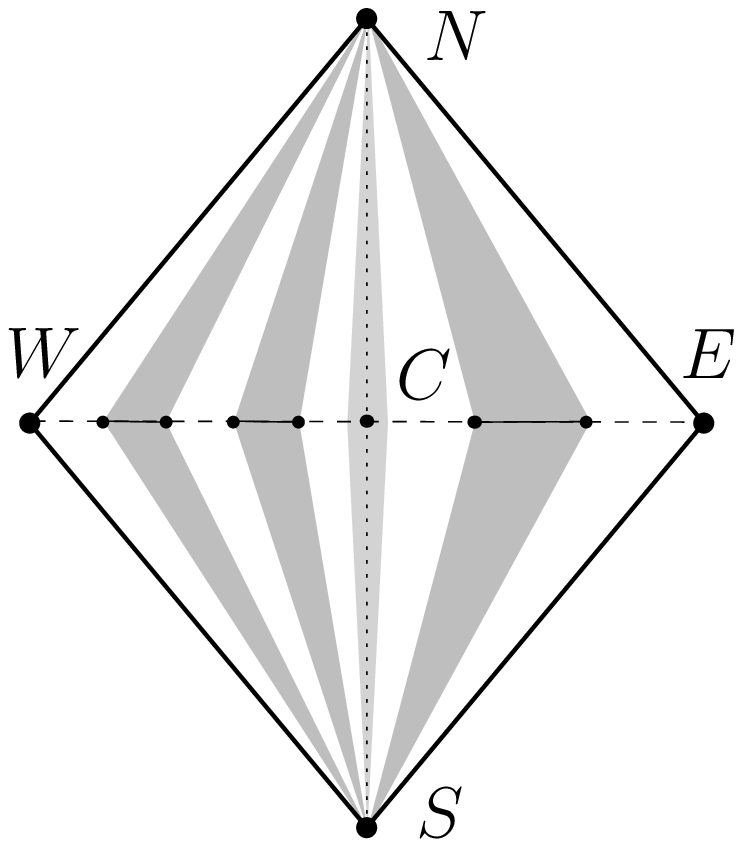}\label{fig:boomerang-e}}
\caption{\small (a) A left boomerang. (b) A right boomerang. (c) A diamond. (d)
Diamonds inside a boomerang. (e)
Boomerangs (and a diamond) inside a diamond.}
\label{fig:boomerang}
\end{figure}

We assign boomerangs (either left or right, depending on the embedding of $G$)
to S-nodes and diamonds to P- and Q-nodes, as follows.

First, consider the Q-node $\rho$ corresponding to the root edge $e$ of $G$. Draw edge $e$ as a segment between points $(0,1)$ and $(0,-1)$. Also, if $\rho$ is adjacent to an S-node $\mu$, then assign to $\mu$ the left boomerang $N=(0,1), E=(-1,0), S=(0,-1), W=(-2,0)$ or the right boomerang $N=(0,1), E=(2,0),
S=(0,-1), W=(1,0)$, depending on the embedding of $G$; if $\rho$ is adjacent to a P-node $\mu$, then associate to $\mu$ the diamond $N=(0,1), E=(+2,0), S=(0,-1), W=(-2,0)$.

Then, consider each node $\mu$ of $T_e(G)$ according to a top-down traversal.

If $\mu$ is an S-node (see Fig.~\ref{fig:boomerang-d}), let $N, E, S, W$ be the boomerang associated with it and let $\alpha$ be the angle $\widehat{WNE}$. We associate diamonds to the children $\mu_1, \mu_2, \dots, \mu_k$ of $\mu$ as follows. Consider the midpoint $C$ of segment $\overline{WE}$. Subdivide $\overline{NC}$ into $\lceil\frac{k}{2}\rceil$ segments with the same length and $\overline{CS}$ into
$\lfloor\frac{k}{2}\rfloor$ segments with the same length. Enclose each of such
segments $\overline{N_i S_i}$, for $i=1, \dots, k$, into a diamond $N_i,
E_i, S_i,W_i$, with $\widehat{W_i N_i E_i} = \alpha$, and associate it with
child $\mu_i$ of $\mu$.

If $\mu$ is a P-node (see Fig.~\ref{fig:boomerang-e}), let $N, E, S,W$ be the
diamond associated with it.
Associate boomerangs and diamonds to the
children $\mu_1, \mu_2, \dots, \mu_k$ of $\mu$ as follows. If a child $\mu_l$ of $\mu$ is a Q-node, then left boomerangs are associated to
$\mu_1, \dots, \mu_{l-1}$, right boomerangs are associated to $\mu_{l+1}, \dots, \mu_k$, and a diamond is associated to $\mu_l$. Otherwise,
right boomerangs are associated to all of $\mu_1, \mu_2, \dots, \mu_k$. We
assume that a child $\mu_l$ of $\mu$ that is a Q-node exists, the description
for the case in which no child of $\mu$ is a Q-node being similar and simpler.
We describe how to associate left boomerangs to the children $\mu_1, \mu_2,
\dots, \mu_{l-1}$ of $\mu$. Consider the midpoint $C$ of segment $\overline{WE}$
and consider $2l$ equidistant points $W=p_1, \dots, p_{2l} = C$ on segment
$\overline{WC}$. Associate each child $\mu_i$, with $i = 1, \dots, l-1$, to the
quadrilateral $N_i = N, E_i= p_{2i}, S_i = S, W_i = p_{2i+1}$. Right boomerangs
are associated to $\mu_{l+1}, \mu_{l+2}, \dots, \mu_k$ in a symmetric way.
Finally, associate $\mu_l$ to any diamond such that $N_l = N, S_l = S$, $W_l$
is any point between $C$ and $E_{l-1}$, and $E_l$ is any point
between $C$ and $W_{l+1}$.

If $\mu$ is a Q-node, let $N, E, S, W$ be the diamond associated with it. Draw
the edge corresponding to $\mu$ as a straight-line segment between $N$
and $S$.

Observe that the above described algorithm constructs a drawing of $G$, that we
call the {\em canonical drawing} of $G$. We now argue that no two edges $e_1$
and $e_2$ intersect in the canonical drawing of $G$. Consider the lowest common
ancestor $\nu$ of the Q-nodes $\tau_1$ and $\tau_2$ of $T_e$ representing $e_1$
and $e_2$, respectively. Also, consider the children $\nu_1$ and $\nu_2$ of
$\nu$ such that the subtree of $T_e$ rooted at $\nu_i$ contains $\tau_i$, for
$i=1,2$. Such children are associated with internally-disjoint regions of the
plane. Since the subgraphs $G_1$ and $G_2$ of $G$ corresponding to $\nu_1$ and
$\nu_2$, respectively, are entirely drawn inside such regions, it follows that
$e_1$ and $e_2$ do not intersect except, possibly, at common endpoints.

In order to construct a pseudo-morph of a straight-line planar drawing
$\Gamma(G)$ of $G$ into its canonical drawing $\Gamma^*(G)$, we do the
following: (i) We perform a contraction of a vertex $v$ of $G$ into a neighbor
of $v$, hence obtaining a drawing $\Gamma(G')$ of a graph $G'$ with $n-1$
vertices; (ii) we inductively construct a pseudo-morph from $\Gamma(G')$ to the
canonical drawing $\Gamma^*(G')$ of $G'$; and (iii) we uncontract $v$ and
perform a sequence of morphing steps to transform $\Gamma^*(G')$ into the
canonical drawing $\Gamma^*(G)$ of~$G$.

We describe the three steps in more detail.

\subsection{Step 1: Contract a Vertex $v$}\label{app:step-1}

Let $T_e(G)$ be the decomposition tree of $G$ rooted at some edge $e$ incident
to the outer face of $G$. Consider a P-node $\nu$ such that the subtree of
$T_e(G)$ rooted at $\nu$ does not contain any other P-node. This implies that
all the children of $\nu$, with the exception of at most one Q-node, are S-nodes
whose children are Q-nodes. Hence, the series-parallel graph $G(\nu)$ associated
to $\nu$ is composed of a set of paths connecting its poles $s$ and $t$. Let
$p_1$ and $p_2$ be two paths joining $s$ and $t$ and such that their union
is a cycle $\cal C$ not containing other vertices in its interior (see
Fig.~\ref{fig:paths-a}). Such paths exist given that the ``rest of the graph''
with respect to $\nu$ is in the outer face of $G(\nu)$, given that the root $e$
of $T_e(G)$ is incident to the outer face of $G$. Internally triangulate $\cal
C$ by adding dummy edges (dashed edges of Fig.~\ref{fig:paths}). Cycle $\cal C$
and the added dummy edges yield a drawing of a biconnected outerplane graph $O$
which, hence, has at least two vertices of degree two.

\begin{figure}[htb]
\begin{center}
\subfigure[]{\includegraphics[scale=1.1]{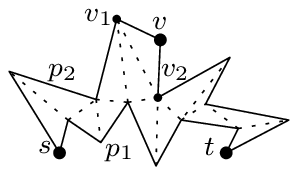}\label{fig:paths-a}}
\hspace{1em}
\subfigure[]{\includegraphics[scale=1.1]{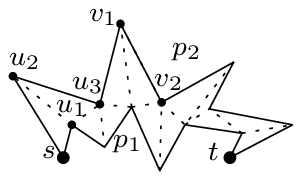}\label{fig:paths-b}}
\hspace{1em}
\subfigure[]{\includegraphics[scale=1.1]{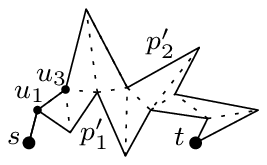}\label{fig:paths-c}}
\caption{The internally triangulated cycle $\cal C$ formed by paths $p_1$ and
$p_2$.
Dummy edges are drawn as dashed lines. (a--b) Vertex $v$ of degree $2$ can be
contracted onto $v_1$. (b--c) Vertex $u_2$ of degree $3$ can be
contracted onto $u_1$.\label{fig:paths}}
\end{center}
\end{figure}

We distinguish two cases depending on the existence of a degree-$2$ vertex $v$
different from $s$ and $t$.

\begin{description}

\item[{Case 1} (there exists a vertex $v$ of degree $2$ different from $s$ and
$t$).]
Assume, without loss of generality, that $v$ belongs to $p_2$. Since $O$ is
internally triangulated, both the neighbors $v_1$ and $v_2$ of $v$ belong to
$p_2$, and they are joined by a dummy edge. We obtain $\Gamma(G')$ from
$\Gamma(G)$ by contracting $v$ onto one of its neighbors, while preserving
planarity (see Figs.~\ref{fig:paths-a} and~\ref{fig:paths-b}).
Either $p_2$ contains more than two edges (\emph{Case 1.1}) or $p_2$ consists of
exactly two edges, namely $(v_1,v)$ and $(v,v_2)$. If the latter case holds,
either edge $(v_1,v_2)$ exists in $G$ (\emph{Case 1.2}) or not (\emph{Case
1.3}). In the three cases we do the following.

\begin{description}
\item[\emph{Case 1.1})] Path $p_2$ is replaced in $G'$ with a path $p_2'$ that
contains edge $(v_1,v_2)$ and does not contain vertex $v$.

\item[\emph{Case 1.2})] Graph $G'$ is set as $G \setminus \{v\}$.

\item[\emph{Case 1.3})] Path $p_2$ is replaced in $G'$ with edge~$(v_1,v_2)$.
\end{description}

\item[{Case 2} (the only two vertices of degree $2$ in $O$ are $s$ and $t$).]
In this case, one of the two vertices $u_1$ and $u_2$ adjacent to $s$ has degree
$3$, say $u_2$ (since the removal of $s$ and its incident edges would yield
another biconnected outerplane
graph with two vertices of degree $2$, namely $t$ and one of $u_1$ and
$u_2$). We obtain $\Gamma(G')$ from $\Gamma(G)$ by contracting $u_2$ onto
$u_1$. Let $u_3$ be the neighbor of $u_1$ and $u_2$ different from $s$.
Since the edges incident to $u_2$ are contained into triangles
$\triangle_{s,u_1,u_2}$
and $\triangle_{u_1, u_2,u_3}$ during the contraction, planarity is
preserved (see Figs.~\ref{fig:paths-b} and~\ref{fig:paths-c}). Let $p_2'$ be the
path composed of edge $(u_1,u_3)$ and of the subpath of $p_2$ between $u_3$ and
$t$, and let $p_1'$ be the subpath of $p_1$ between $u_1$ and $t$.
Observe that $G'$ contains edge $(u_1,u_3)$ and does not contain vertex $u_2$.
\end{description}

Tree $T_e(G')$ is obtained from $T_e(G)$ by performing the
local changes described hereunder, with respect to the above cases.

\begin{description}

\item[{Case 1}.] Let $\tau_1$ and $\tau_2$
be the nodes of $T_e(G)$ corresponding to paths $p_1$ and $p_2$. Note that
$\tau_2$ is an S-node, as $v \in p_2$ and $v \neq s,t$. The two Q-nodes
that are children of $\tau_2$ and that correspond to edges $(v,v_1)$ and
$(v,v_2)$ are removed in $T_e(G')$.

\begin{figure}[htb]
\begin{center}
\subfigure[]{\includegraphics[scale=0.5]{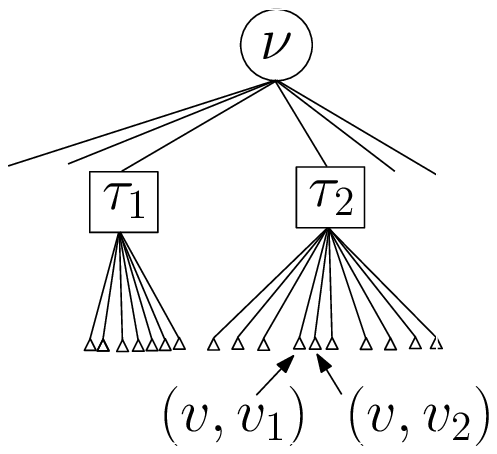}\label{fig:sp-tree1}}
    \hfill
\subfigure[]{\includegraphics[scale=0.5]{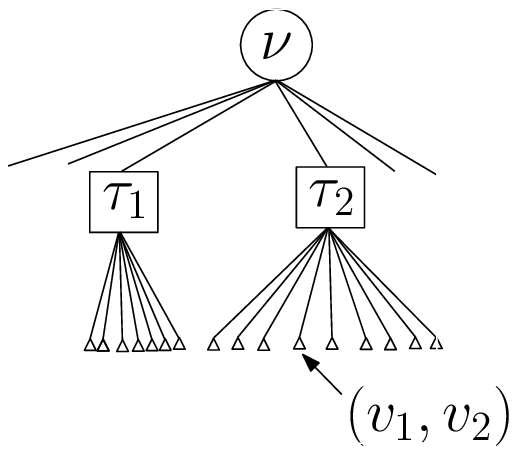}\label{fig:sp-tree2}}
\hfill
\subfigure[]{\includegraphics[scale=0.5]{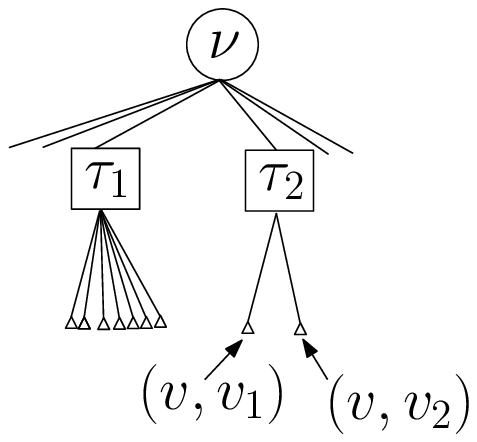}\label{fig:sp-tree3}}
\hfill
\subfigure[]{\includegraphics[scale=0.5]{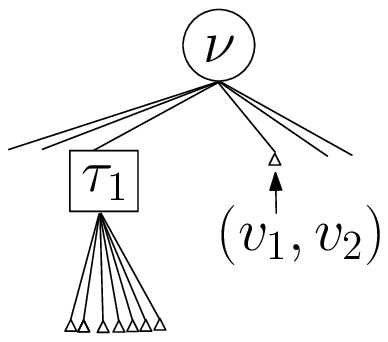}\label{fig:sp-tree4}}
\caption{Construction of $T_e(G')$ starting from $T_e(G)$ in \emph{Case 1}.
(a--b) $T_e(G)$ and $T_e(G')$, respectively, in \emph{Case 1.1}. (c--d) $T_e(G)$
and $T_e(G')$, respectively, in \emph{Case 1.3}. \label{fig:sp-tree}}
\end{center}
\end{figure}

\begin{description}
\item[\emph{Case 1.1})] A Q-node corresponding to $(v_1,v_2)$ is added as
a child of $\tau_2$ (see Figs.~\ref{fig:sp-tree1} and~\ref{fig:sp-tree2}).

\item[\emph{Case 1.2})] $\tau_2$ is removed from $T_e(G')$. Also, if $\nu$ has
no
children other than $\tau_1$ and $\tau_2$ in $T_e(G')$, then $\nu$ is replaced
with $\tau_1$ in $T_e(G')$.

\item[\emph{Case 1.3})] $\tau_2$ is replaced in $T_e(G')$ with a Q-node
corresponding to $(v_1,v_2)$ (see
Figs.~\ref{fig:sp-tree3} and~\ref{fig:sp-tree4}).
\end{description}

\item[{Case 2}.] Let $\tau_1$
and $\tau_2$ be the nodes of $T_e(G)$ corresponding to paths $p_1$ and $p_2$,
and let $\mu$ be the parent of $\nu$. Note that $\tau_1$ and $\tau_2$ are
S-nodes, as $u_1,u_2 \neq s,t$. First, the Q-nodes corresponding to edges
$(s,u_2)$ and $(u_2,u_3)$ are removed from the children of $\tau_2$, and a
Q-node $\nu_Q$ (corresponding to edge $(u_1,u_3)$) is added to $T_e(G')$.
We distinguish the cases in which $\nu$ has more than two children in $T_e(G)$
(\emph{Case 2.1})
and when $\nu$ has exactly two children in $T_e(G)$ (\emph{Case 2.2}).

\begin{description}
\item[\emph{Case 2.1})] An
S-node $\nu_S$ and a P-node $\nu_P$ are introduced in $T_e(G')$, in such a way
that (i) $\nu_S$ is a child of $\nu$, (ii) the Q-node corresponding to $(s,u_1)$
and $\nu_P$ are children of $\nu_S$, (iii) $\tau_1$ and $\tau_2$ are children of
$\nu_P$, and (iv) $\nu_Q$ is a child of $\tau_2$. See Figs.~\ref{fig:sp-tree5}
and~\ref{fig:sp-tree6}.

\begin{figure}[htb]
\begin{center}
\subfigure[]{\includegraphics[scale=0.45]{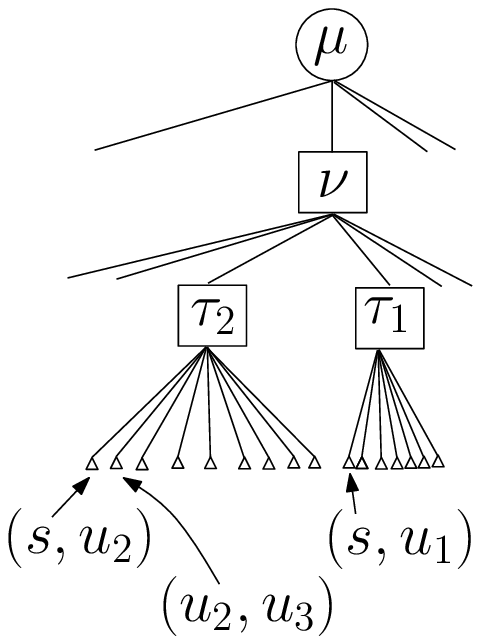}\label{fig:sp-tree5}}
\hfill
\subfigure[]{\includegraphics[scale=0.45]{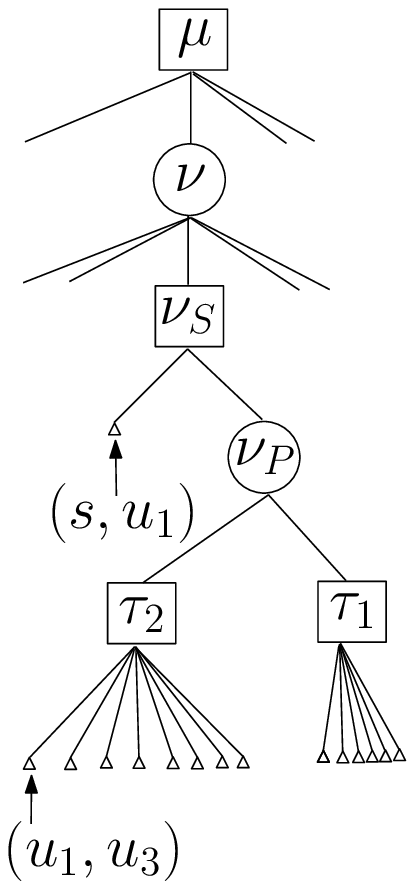}\label{fig:sp-tree6}}
\hfill
\subfigure[]{\includegraphics[scale=0.45]{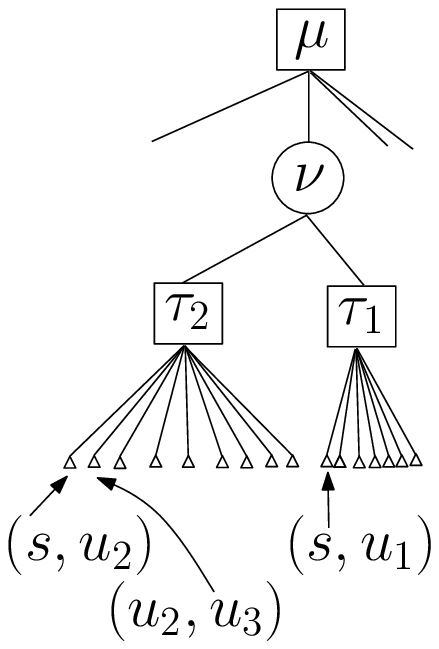}\label{fig:sp-tree7}}
\hfill
\subfigure[]{\includegraphics[scale=0.45]{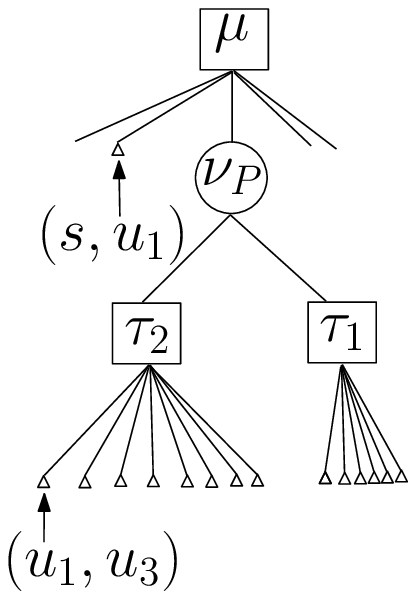}\label{fig:sp-tree8}}
\caption{Construction of $T_e(G')$ starting from $T_e(G)$ in \emph{Case 2}.
(a--b) $T_e(G)$ and $T_e(G')$, respectively, in \emph{Case 2.1}. (c--d) $T_e(G)$
and $T_e(G')$, respectively, in \emph{Case 2.2}.\label{fig:sptreebis}}
\end{center}
\end{figure}

\item[\emph{Case 2.2})] Node $\nu$ is
removed from the children of $\mu$, and a P-node $\nu_P$ is introduced in
$T_e(G')$ in such a way that (i) the Q-node corresponding to $(s,u_1)$ and
$\nu_P$ are children of $\mu$, (ii) $\tau_1$ and $\tau_2$ are children of
$\nu_P$, and (iii) $\nu_Q$ is a child of $\tau_2$. See Figs.~\ref{fig:sp-tree7}
and~\ref{fig:sp-tree8}.

\end{description}

\end{description}

\subsection{Step 2: Recursive Call}

Let $\Gamma(G')$ be the drawing of the graph $G' = G \setminus \{v\}$ obtained
after the contraction of vertex $v$ performed in \emph{Case 1} or in
\emph{Case 2}.

Inductively construct a morphing from $\Gamma(G')$ to the canonical
drawing $\Gamma^*(G')$ of $G'$ in $c \cdot (n-1)$ steps, where $c$ is
a constant.

\subsection{Step 3: Uncontract Vertex $v$ and Construct a Canonical Drawing of
$G$}

We describe how to obtain $\Gamma^*(G)$ from $\Gamma^*(G')$ by uncontracting $v$
and performing a constant number of morphing steps. The description follows the
cases discussed in Appendix~\ref{app:step-1}.

\begin{description}

\item[{Case 1} (there exists a vertex $v$ of degree $2$ different from $s$ and
$t$).] ~

\begin{description}

\item[\emph{Case 1.1})] This case is discussed in
Section~\ref{sse:biconnected-sp}.

\item[\emph{Case 1.2} and \emph{Case 1.3})] Note that $\Gamma^*(G')$ and
$\Gamma^*(G)$ coincide, except for the fact that: (i) $\Gamma^*(G)$ contains one
boomerang more than $\Gamma^*(G')$
(the one associated to $\tau_2$) inside the diamond associated to $\nu$, (ii)
$\Gamma^*(G)$ might not contain the diamond associated to the Q-node
corresponding to edge $(s,t)$ (in \emph{Case 1.3}), and (iii) the boomerangs
inside the diamond associated to $\nu$ have a different drawing in
$\Gamma^*(G')$ and $\Gamma^*(G)$. Drawing $\Gamma^*(G')$ is
illustrated in Fig.~\ref{fig:Case12-extraction-1}, drawing $\Gamma^*(G)$ in
\emph{Case 1.2} is illustrated in Fig.~\ref{fig:Case12-extraction-3}, drawing
$\Gamma^*(G)$ in \emph{Case 1.3} is illustrated in
Fig.~\ref{fig:Case12-extraction-4}.

Since edge $(v_1,v_2)$ exists in $G'$, its drawing in $\Gamma^*(G')$ is the
straight-line segment between the points $N'$ and $S'$ of a diamond
$N',E',S',W'$. Also, the drawing $\Gamma^*(p_2)$ of $p_2$ in $\Gamma^*(G)$ lies
inside a boomerang $N,E,S,W$ with $N=N'$ and $S=S'$.

\begin{figure}[htb]\label{fig:Case12}
\centering
\subfigure[]{\includegraphics[scale=0.53]{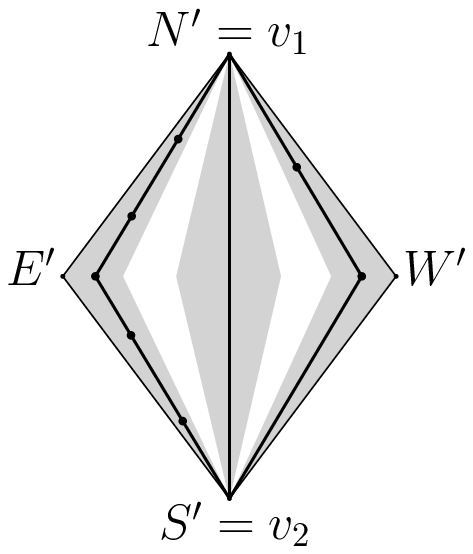}
\label{fig:Case12-extraction-1}}\hspace{.1em}
\subfigure[]{\includegraphics[scale=0.53]{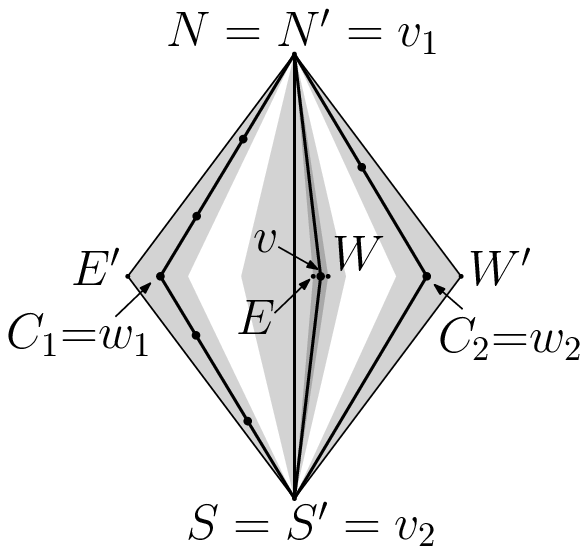}
\label{fig:Case12-extraction-2}}\hspace{.1em}
\subfigure[]{\includegraphics[scale=0.53]{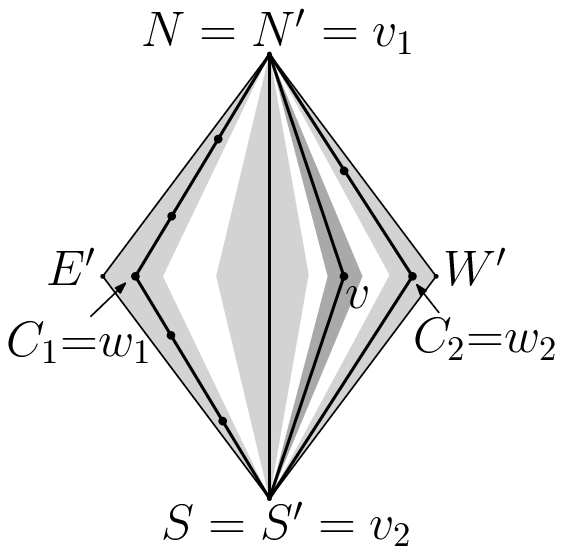}
\label{fig:Case12-extraction-3} }\hspace{.1em}
\subfigure[]{\includegraphics[scale=0.53]{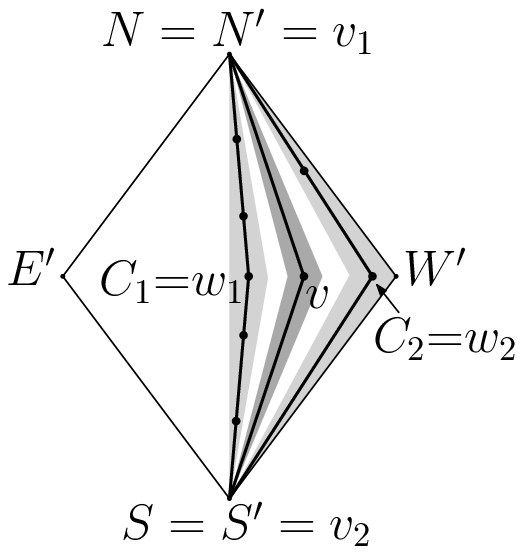}
\label{fig:Case12-extraction-4} }
\caption{Construction of $\Gamma^*(G)$ from $\Gamma^*(G')$ when either
\emph{Case 1.2} or \emph{Case 1.3} applied. (a) $\Gamma^*(G')$. The
diamond associated to edge $(v_1,v_2)$ and the boomerangs associated to
the children $\tau_i$ of $\nu$ are light-grey. (b) Two points $W$ and $E$
are selected on $\overline{E'W'}$, creating a (dark gray) boomerang associated
to $\tau_2$, and vertex $v$ is moved to the midpoint of $\overline{EW}$. (c)
$\Gamma^*(G)$ in \emph{Case 1.2}, where edge $(v_1,v_2)$ exists in $G$. (d)
$\Gamma^*(G)$ in \emph{Case 1.3}.
}
\end{figure}

In order to construct a pseudo-morph from $\Gamma^*(G')$ to $\Gamma^*(G)$,
initially
place points $E$ and $W$ on segment $\overline{E'W'}$, on the same side with
respect to segment $\overline{N'S'}$ (in \emph{Case 1.2}, the side depends on
the order of the children of $\nu$ in $T_e(G)$). With one morphing step, move
$v$ to the midpoint of segment $\overline{EW}$ (see
Fig.~\ref{fig:Case12-extraction-2}).

Consider the children $\tau_i$ of $\nu$ in $T_e(G)$ that
are not Q-nodes, with $i=1,\dots,q$, and note that the drawing of each $\tau_i$
is composed of two straight-line segments $\overline{NC_i}$ and
$\overline{SC_i}$.
With a second morphing step, move the vertex $w_i$ of $\tau_i$ lying on $C_i$,
for each $i=1,\dots,q$, and vertex $v$ along the line through $\overline{EW}$
till reaching their positions in $\Gamma^*(G)$. In the same morphing step, for
each $i=1,\dots,q$, the vertices on the path between $s$ and $w_i$ are moved as
convex combination of the movements of $s$ and $w_i$, and the vertices on the
path between $t$ and $w_i$ are moved as linear combination of the movements of
$t$ and $w_i$. Hence, at the end of the morphing step, also such vertices reach
their positions in $\Gamma^*(G)$ (see
Figs.~\ref{fig:Case12-extraction-3} and~\ref{fig:Case12-extraction-4}).

\end{description}

\item[{Case 2} (the only two vertices of degree $2$ in $O$ are $s$ and $t$).] ~

\begin{description}

\item[\emph{Case 2.1})]
Note that $\Gamma^*(G')$ and
$\Gamma^*(G)$ coincide, except for the drawing of $p_1$, $p_2$, $p_1'$, and
$p_2'$.

Namely, $p_1$ and $p_2$ are drawn in $\Gamma^*(G)$ in two boomerangs
$N_1,E_1,S_1,W_1$ and $N_2=N_1,E_2,S_2=S_1,W_2$ lying inside the diamond
associated to $\nu$ (see Fig.~\ref{fig:Case21-extraction-finale}), while $p_1'$
and $p_2'$ are drawn in $\Gamma^*(G')$ in two boomerangs
$N_1',E_1',S_1'=S_1,W_1'$ and $N_2'=N_1',E_2',S_2'=S_1'=S_1,W_2'$ lying inside a
diamond associated to $\nu_P$, that lies inside a boomerang
$N_S,E_S,S_S,W_S$ associated to $\nu_S$ (with $S_S=S_2'=S_1'=S_1$), that lies
inside the diamond associated to $\nu$ (see
Fig.~\ref{fig:Case21-extraction-1}).

Note that, since $\nu_S$ has two children in $T_e(G')$, vertex $u_1$ is placed
on the midpoint $C_S$ of segment $\overline{E_SW_S}$, that is, $C_S=N_1'=N_2'$.

Let $w_1$ and $w_2$ be the vertices of $p_1'$ and $p_2'$, respectively, placed
on the midpoints $C_1'$ and $C_2'$ of segments $\overline{E_1'W_1'}$ and
$\overline{E_2'W_2'}$.

\begin{figure}[htb]
\centering
\subfigure[]{\includegraphics[scale=0.62]{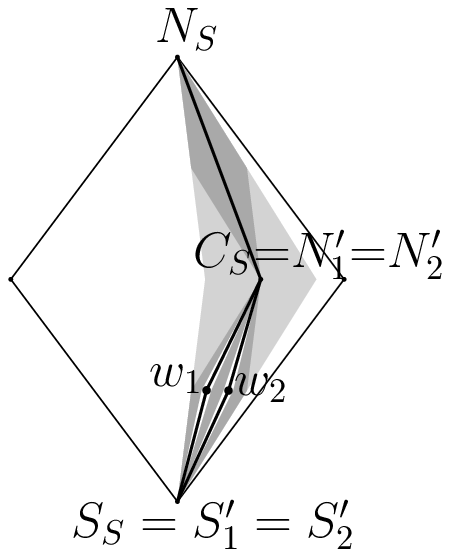}
\label{fig:Case21-extraction-1}}\hfill
\subfigure[]{\includegraphics[scale=0.62]{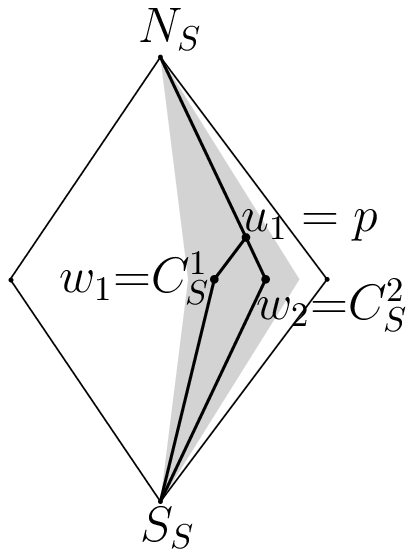}
\label{fig:Case21-extraction-2}}\hfill
\subfigure[]{\includegraphics[scale=0.62]{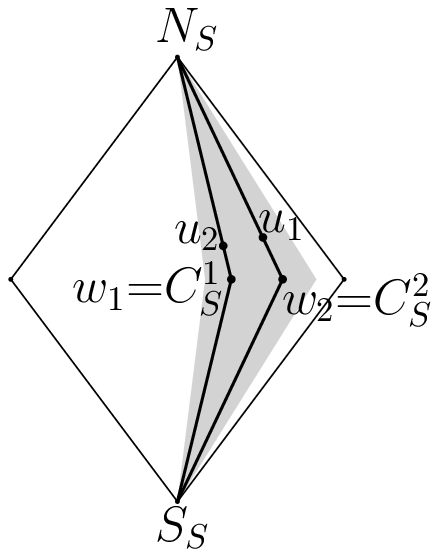}
\label{fig:Case21-extraction-3} }\hfill
\subfigure[]{\includegraphics[scale=0.62]{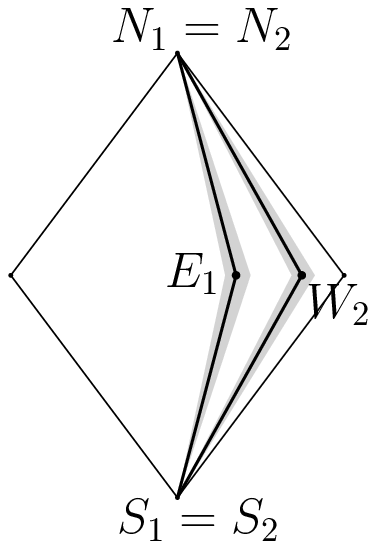}
\label{fig:Case21-extraction-finale}}
\caption{Construction of $\Gamma^*(G)$ from $\Gamma^*(G')$ when \emph{Case
$2.1$} applied. (a) $\Gamma^*(G')$. The boomerang associated to $\nu_S$ is
light-grey, the diamonds associated to $\nu_P$ and to the Q-node corresponding
to $(s,u_2)$ are dark-grey, and the boomerangs associated to $\tau_1$ and
$\tau_2$ are white. (b) Vertices $w_1$ and $w_2$ are moved to points $C_S^1$ and
$C_S^2$, and $u_1$ is moved to a point of $\overline{N_SC_S^2}$. (c) Vertex $v$
is uncontracted from $u_2$ and moved to a point of $\overline{N_SC_S^1}$. (d)
$\Gamma^*(G)$. The boomerangs associated to $\tau_1$ and $\tau_2$ are
light-grey.\label{fig:Case21}}
\end{figure}

With one morphing step, move $w_1$ to any point $C_S^1$ on $\overline{E_SC_S}$,
move $w_2$ to any point $C_S^2$ on $\overline{C_SW_S}$, and move $u_1$ to any
point on segment $\overline{N_S C_S^2}$ (see
Fig.~\ref{fig:Case21-extraction-2}). In the same morphing step, for each
two vertices in $w_1$, $w_2$, $u_1$, and $t$, say $w_1$ and $t$, the vertices
lying on segment $\overline{w_1t}$ are moved as linear combination of the
movements of $w_1$ and $t$. Hence, at the end of the morphing step, all these
vertices still lie on $\overline{w_1t}$.

Next, with one morphing step, uncontract $u_2$ from $u_1$ and move it to any
internal point
on segment $\overline{N_S C_S^1}$ (see Fig.~\ref{fig:Case21-extraction-3}). In
the same morphing step, the vertices lying on segment $\overline{u_2 w_1}$ are
moved as linear combination of the movements of $u_2$ and $w_1$.

Further, perform the same operation as in \emph{Case 1.1} to redistribute the
vertices of $p_1$ on $\overline{N_SC_S^1}$ and $\overline{S_SC_S^1}$, and the
vertices of $p_2$ on $\overline{N_SC_S^2}$ and $\overline{S_SC_S^2}$. After
this step, for each child $\tau_i$ of $\nu$, the vertex $w_i$ of $\tau_i$
lying on segment $\overline{E_SW_S}$ in $\Gamma^*(G)$ lies on
$\overline{E_SW_S}$ also in the current drawing.

Finally, perform the same operation as in \emph{Case 1.2} to move the
vertex $w_i$ of each child $\tau_i$ of $\nu$ to its final position (on segment
$\overline{E_SW_S}$) in $\Gamma^*(G)$. In the same morphing step, the vertices
on the path between $s$ and $w_i$ are moved as linear combination of the
movements of $s$ and $w_i$, and the vertices on the path between $t$ and $w_i$
are moved as linear combination of the movements of $t$ and $w_i$. Hence, at the
end of the morphing step, also such vertices reach their positions in
$\Gamma^*(G)$.

\item[\emph{Case 2.2})] Note that $\Gamma^*(G')$ and
$\Gamma^*(G)$ coincide, except for the drawing of $p_1$, $p_2$, $p_1'$, and
$p_2'$.

Namely, $p_1$ and $p_2$ are drawn in $\Gamma^*(G)$ in two boomerangs
(associated to $\tau_1$ and $\tau_2$) lying inside the diamond associated to
$\nu$ (see Fig.~\ref{fig:Case22-extraction-finale}), that lies inside the
boomerang associated to $\mu$. Also $p_1'$ and $p_2'$ are drawn in
$\Gamma^*(G')$ in two boomerangs lying inside a diamond (associated to $\nu_P$)
that lies inside the boomerang associated to $\mu$ (see
Fig.~\ref{fig:Case22-extraction-1}). However, the boomerang associated to $\mu$
in $\Gamma^*(G)$ has one diamond less than in $\Gamma^*(G')$, since in
$\Gamma^*(G')$ it also contains the diamond associated to edge $(s,u_1)$. Also,
the vertices in the boomerangs associated to $\tau_1$ and $\tau_2$ have
different positions in $\Gamma^*(G')$ and in $\Gamma^*(G)$, since vertex $u_2$
is not present in $\Gamma^*(G')$.

With three morphing steps analogous to those performed in \emph{Case 1.1}, we
redistribute the vertices inside the boomerang $N,E,S,W$ associated to $\mu$ in
such a way that the vertex lying on the midpoint $C$ of $\overline{EW}$ is the
same in $\Gamma^*(G')$ and in $\Gamma^*(G)$. Note that, after these steps, the
diamonds associated to $\nu_P$ and to edge $(s,u_1)$ lie on the
same segment, either $\overline{NC}$ or $\overline{SC}$, say $\overline{SC}$,
and that the vertices lying on segment $\overline{NC}$ already are at their
final position in $\Gamma^*(G)$. Then, with three morphing steps analogous to
those performed in \emph{Case 2.1}, we uncontract $u_2$ and collapse the two
diamonds associated to $\nu_P$ and to $(s,u_1)$ into a single diamond. Then,
with one morphing step (analogous to one of the steps performed in \emph{Case
$1.1$}), we move the vertices lying on segment $\overline{SC}$ till they
reach their final position in $\Gamma^*(G)$.

\begin{figure}[htb]\label{fig:Case22}
\centering
\subfigure[]{\includegraphics[scale=0.62]{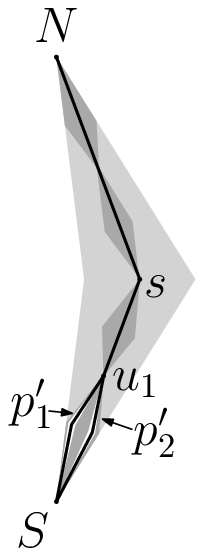}
\label{fig:Case22-extraction-1}}\hspace{3em}
\subfigure[]{\includegraphics[scale=0.62]{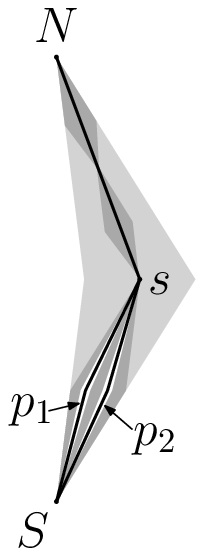}
\label{fig:Case22-extraction-finale}}
\caption{Construction of $\Gamma^*(G)$ from $\Gamma^*(G')$ when \emph{Case
$2.2$} applied. (a) $\Gamma^*(G')$. The boomerang associated to $\mu$ is
light-grey, the diamonds associated to the children of $\mu$, including $\nu_P$
and the Q-node corresponding to $(s,u_1)$ are dark-grey, and the boomerangs
associated to $\tau_1$ and $\tau_2$ are white. (b) $\Gamma^*(G)$.}
\end{figure}

\end{description}

\end{description}

\subsection{Simply-Connected Series-Parallel
Graphs}\label{sse:simply-connected-sp}

In this section we show how, by preprocessing the input drawings $\Gamma_a$ and
$\Gamma_b$ of any series-parallel graph $G$, the algorithm presented in
\Cref{sse:biconnected-sp} can be used to compute a pseudo-morph $M=$\mmorph{\Gamma_a,\dots,\Gamma_b}. The idea is to augment both $\Gamma_a$ and
$\Gamma_b$ to two drawings $\Gamma_a'$ and $\Gamma_b'$ of a biconnected
series-parallel graph $G'$, compute the morph
$M'=$\mmorph{\Gamma'_a,\dots,\Gamma'_b}, and obtain $M$ by restricting $M'$ to
the vertices and edges of $G$.

This augmentation is performed on $G$ by repeatedly applying the following
lemma.

\begin{lemma}\label{le:simplebico}
Let $v$ be a cut-vertex of a plane series-parallel graph $G$ with $n_b$
blocks. Let $e_1=(u,v)$ and $e_2=(w,v)$ be two consecutive edges in the circular
order around $v$ such that $e_1$ belongs to block $b_1$ of
$G$ and $e_2$ belongs to block $b_2 \neq b_1$ of $G$.
The graph $G^*$ obtained from $G$ by adding a vertex $z$ and edges $(u,z)$ and
$(w,z)$ is a plane series-parallel graph with $n_b -1$ blocks.
\end{lemma}

\begin{proof}
First, observe that by adding $z$, $(u,z)$, and $(w,z)$ to $G$, blocks $b_1$
and $b_2$ are merged together into a single block $b_{1,2}$ of $G^*$ (see
Figs.~\ref{fig:block_simply} and~\ref{fig:edges_simply}). Hence, the number of
blocks of $G^*$ is $n_b - 1$. It remains to show that $G^*$ is a
series-parallel graph.

\begin{figure}[h]
\centering
~\hfill
\subfigure[]{\includegraphics[scale=.5]{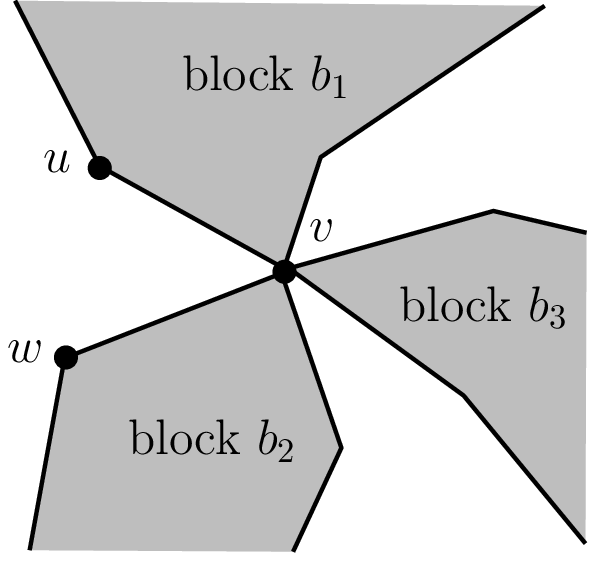}\label{fig:block_simply}}
\hfill
\hfill
\subfigure[]{\includegraphics[scale=.5]{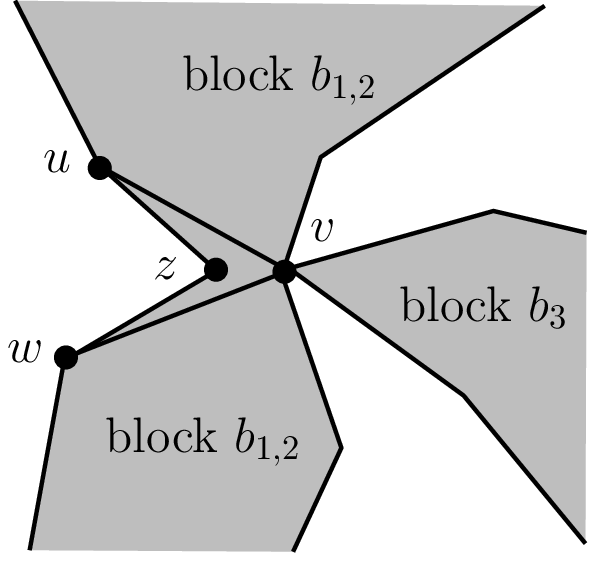}\label{fig:edges_simply}}
\hfill~
 \caption{(a) A cut-vertex $v$ of a series-parallel graph $G$.
(b) Graph $G'$ obtained by adding vertex $z$ and edges  $(u,z)$, and $(w,z)$ is
a series-parallel graph.\label{fig:simply}}
\end{figure}

Assume for a contradiction that $G^*$ is not a series-parallel graph.
It follows that $G^*$ contains a subdivision of the complete graph on four
vertices $K_4$, i.e., there is a set $V_{K_4}$ of four vertices of $G^*$ such
that any two of them are joined by three vertex-disjoint paths. Observe that the
vertices in $V_{K_4}$ cannot belong to different blocks of $G^*$. Further, since
$G$ is a series-parallel graph, the vertices in $V_{K_4}$ belong to $b_{1,2}$.
Since $z$ has degree two, $z \notin V_{K_4}$; hence, the vertices in $V_{K_4}$
are also vertices of $G$. This gives a contradiction since: (i) The vertices in
$V_{K_4}$ cannot all belong to $b_1$, as otherwise $G$ would not be
series-parallel, contradicting the hypothesis; (ii) the vertices in $V_{K_4}$
cannot all belong to $b_2$, for the same reason; and (iii) the vertices in
$V_{K_4}$ cannot belong both to $b_1$ and $b_2$, as otherwise there could
not exist three vertex-disjoint paths joining them in $G^*$, contradicting the
hypothesis that $G^*$ contains a subdivision of $K_4$.
\end{proof}

Observe that, when augmenting $G$ to $G^*$, both $\Gamma_a$ and $\Gamma_b$ can
be augmented to two planar straight-line drawings $\Gamma^*_a$ and $\Gamma^*_b$
of $G^*$ by placing vertex $z$ suitably close to $v$ and with direct visibility
to vertices $u$ and $w$, as in the proof of F\'{a}ry's
Theorem~\cite{fary-slrpg-48}.
By repeatedly applying such an augmentation we obtain a biconnected
series-parallel graph $G'$ and its drawings $\Gamma_a'$ and $\Gamma_b'$, whose
number of vertices and edges is linear in the size of $G$. Hence,
the algorithm described in \Cref{sse:biconnected-sp} can be applied to obtain
a pseudo-morph \mmorph{\Gamma_a,\dots,\Gamma_b}, thus proving
\Cref{th:sp-pseudo-morphing}. We will show in Section~\ref{se:geometry} how to
obtain a morph starting from the pseudo-morph computed in this section.

\section{Morphing Plane Graph Drawings in $\mathbf{O(n^2)}$
Steps}\label{se:general_case}

In this section we prove the following theorem.

\begin{theorem}\label{th:pseudo-general_case}
Let $\Gamma_s$ and $\Gamma_t$ be two drawings of the same plane graph $G$.
There exists a pseudo-morph \mmorph{\Gamma_s,\dots,\Gamma_t} with $O(n^2)$
steps transforming $\Gamma_s$ into $\Gamma_t$ .
\end{theorem}

\subsubsection{Preliminary definitions}
Let $\Gamma$ be a planar straight-line drawing of a plane graph $G$. A face $f$
of $G$ is \emph{empty} in $\Gamma$ if it is delimited by a simple cycle.
Consider a vertex $v$ of $G$ and let $v_1$ and $v_2$ be two of its neighbors.
Vertices $v_1$ and $v_2$ are \emph{consecutive} neighbors of $v$ if no edge
appears between edges $(v,v_1)$ and $(v,v_2)$ in the circular order of the edges
around $v$ in $\Gamma$.
Let $v$ be a vertex with $\deg(v) \leq 5$ such that each face containing
$v$ on its boundary is empty. We say that $v$ is \emph{contractible}
\cite{aacdfl-mpgdpns-13-c} if, for each two neighbors $u_1$ and $u_2$ of $v$,
edge $(u_1,u_2)$ exists in $G$ \emph{if and only if} $u_1$ and $u_2$ are
consecutive neighbors of $v$. We say that $v$ is \emph{quasi-contractible} if, for
each two neighbors $u_1$ and $u_2$ of $v$, edge $(u_1,u_2)$ exists in $G$
\emph{only if} $u_1$ and $u_2$ are consecutive neighbors of $v$. In other words,
no edge exists between non-consecutive neighbors of a contractible or quasi-contractible vertex; also, each face incident to a contractible vertex $v$ is delimited by a $3$-cycle, while a face incident to a quasi-contractible vertex might have more than three incident vertices. We have the following.

\begin{lemma}\label{le:qcv_exists}
Every planar graph contains a \qcv vertex.
\end{lemma}

\begin{proof}
Let $\Gamma$ be a planar drawing of a graph $G$.
Add vertices $a$, $b$,
and $c$ so that the triangle composed by these vertices completely encloses
$\Gamma$, and augment the obtained drawing to the drawing $\Gamma'$ of a maximal
plane graph $G'$ by adding dummy edges. Since $G'$ is maximal plane, it contains
a \emph{contractible} vertex $v$ (different from $a$, $b$, and $c$), as shown in
\cite{aacdfl-mpgdpns-13-c}. Since $v$ is contractible in $G'$, it is either
contractible or \qcv in $G$, as the edges incident to a vertex in $G \cap G'$
are at most those of $G'$.
\end{proof}

Further, given a neighbor $x$ of $v$, we say that $v$ is \emph{$x$-contractible}
onto $x$ in $\Gamma$ if:
\begin{inparaenum}[$(i)$]
    \item $v$ is quasi-contractible, and
    \item the contraction of $v$ onto $x$ in $\Gamma$ results in a
straight-line planar drawing $\Gamma'$ of $G'=G / (v,x)$.
\end{inparaenum}

\subsubsection{The algoritm}
We describe the main steps of our algorithm to pseudo-morph a drawing
$\Gamma_s$ of a plane graph $G$ into another drawing $\Gamma_t$ of $G$.

First, we consider a quasi-contractible vertex $v$ of $G$, that exists by Lemma
\ref{le:qcv_exists}.
Second, we compute a pseudo-morph with $O(n)$ steps of $\Gamma_s$ into a drawing $\Gamma^x_s$ of $G$ and a pseudo-morph with $O(n)$ steps of $\Gamma_t$ into a drawing $\Gamma^x_t$ of $G$, such that $v$ is $x$-contractible onto the same neighbor $x$ both in $\Gamma^x_s$ and in $\Gamma^x_t$. We will describe later how to perform these pseudo-morphs.
Third, we contract $v$ onto $x$ both in $\Gamma^x_s$ and in $\Gamma^x_t$, hence obtaining two
drawings $\Gamma_s'$ and $\Gamma_t'$ of a graph $G' = G / (v,x)$ with $n-1$
vertices.
Fourth, we recursively compute a pseudo-morph transforming $\Gamma_s'$ into $\Gamma_t'$.  This completes the description of the algorithm for constructing a pseudo-morphing transforming $\Gamma_s$ into $\Gamma_t$. Observe that the algorithm has $p(n) \in O(n^2)$ steps, thus proving Theorem~\ref{th:pseudo-general_case}. Namely, as it will be described later, $O(n)$ steps suffice to construct pseudo-morphings of $\Gamma_s$ and $\Gamma_t$ into drawings $\Gamma^x_s$ and $\Gamma^x_t$ of $G$, respectively, such that $v$ is $x$-contractible onto the same neighbor $x$ both in $\Gamma^x_s$ and in $\Gamma^x_t$. Further, two steps are sufficient to contract $v$ onto $x$ in both $\Gamma^x_s$ and $\Gamma^x_t$, obtaining drawings $\Gamma_s'$ and $\Gamma_t'$, respectively. Finally, the recursion on $\Gamma_s'$ and $\Gamma_t'$ takes $p(n-1)$ steps. Thus, $p(n)=p(n-1)+O(n) \in O(n^2)$. We will show in Section~\ref{se:geometry} how to obtain a morph starting from the pseudo-morph computed in this section.

We remark that our approach is similar to the one proposed by Alamdari \emph{et al.}~\cite{aacdfl-mpgdpns-13-c}. In~\cite{aacdfl-mpgdpns-13-c}
$\Gamma_s$ and $\Gamma_t$ are augmented to drawings of the same maximal planar graph with $m\in O(n^2)$ vertices; then, Alamdari \emph{et al.} show how to construct a morphing in $O(m^2)$ steps between two drawings of the same $m$-vertex maximal planar graph. This results in a morphing between $\Gamma_s$ and $\Gamma_t$ with $O(n^4)$ steps. Here, we also augment $\Gamma_s$ and $\Gamma_t$ to drawings of maximal planar graphs. However, we only require that the two
maximal planar graphs coincide in the subgraph induced by the neighbors of $v$. Since this can be achieved by adding a constant number of vertices to $\Gamma_s$ and $\Gamma_t$, namely one for each of the at most five faces $v$ is incident to, our morphing algorithm has $O(n^2)$ steps.

\subsubsection{Making $v$ $x$-contractiblle}\label{sse:algo_gencase}

Let $v$ be a quasi-contractible vertex of $G$. We show an algorithm to construct a pseudo-morph with $O(n)$ steps transforming any straight-line planar drawing $\Gamma$ of $G$ into a straight-line planar drawing $\Gamma'$ of $G$ such that $v$ is $x$-contractible onto any neighbor $x$. If $v$ has degree $1$, then it is contractible into its unique neighbor in $\Gamma$, and there is nothing to prove.

\remove{
\begin{figure}[b]
\centering
\includegraphics[scale=.8]{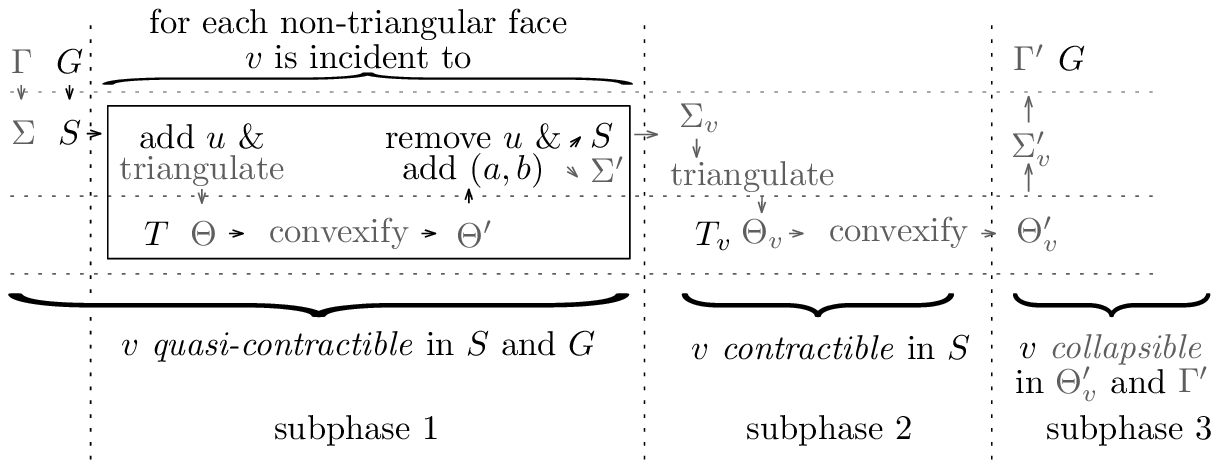}\label{fig:gencase}
\caption{Outline of the first phase of the algorithm, transforming a drawing
$\Gamma$ of $G$ into $\Gamma'$ where vertex $v$ is $x$-contractible. Rows refer
to different (support) graphs, while columns to different sub-phases. Grey
identifies operations on drawings; black those on graphs.}
\end{figure}
}

In order to transform $\Gamma$ into $\Gamma'$, we use a support graph $S$ and
its drawing $\Sigma$, initially set equal to $G$ and $\Gamma$, respectively.
The goal is to augment $S$ and $\Sigma$ so that $v$ becomes a contractible vertex of
$S$. In order to do this, we have to add to $S$ an edge between any two consecutive neighbors of $v$. However, the
insertion of these edges might not be possible in $\Sigma$, as it might lead to a crossing or to enclose some vertex inside a cycle delimited by $v$ and by two consecutive neighbors of $v$ (see \Cref{fig:v_noedge}).

Let $a$ and $b$ be two consecutive neighbors of $v$. If the closed triangle
\mpoly{a,b,v} does not contain any vertex other than $a$, $b$, and $v$, then add
edge $(a,b)$ to $S$ and to $\Sigma$ as a straight-line segment.
Otherwise, proceed as follows. Let $\Sigma_u$ be the drawing of a plane graph
$S_u$ obtained by adding a vertex $u$ and the edges $(u,v)$, $(u,a)$, and
$(u,b)$ to $\Sigma$ and to $S$, in such a way that the resulting drawing is
straight-line planar and each face containing $u$ on its boundary is empty.
As in the proof of F\'{a}ry's Theorem \cite{fary-slrpg-48}, a position for $u$
with such properties can be found in $\Sigma$, suitably close to $v$. See
\Cref{fig:v_dummy} for an example.

\begin{figure}[htb]
\centering
\subfigure[]{\includegraphics[scale=.8]{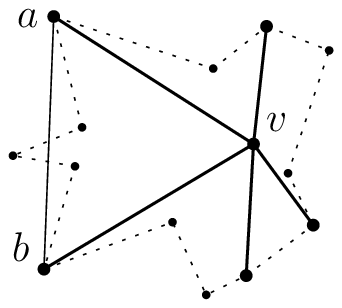}\label{fig:v_noedge}}\hspace{.8em}
\subfigure[]{\includegraphics[scale=.8]{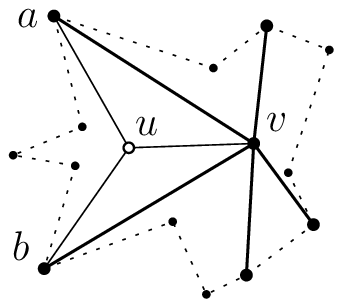}\label{fig:v_dummy}}\hspace{.8em}
\subfigure[]{\includegraphics[scale=.8]{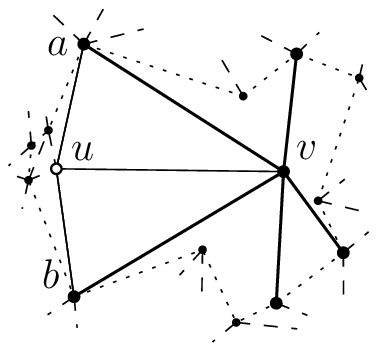}\label{fig:v_convex}}\hspace{.8em}
\subfigure[]{\includegraphics[scale=.8]{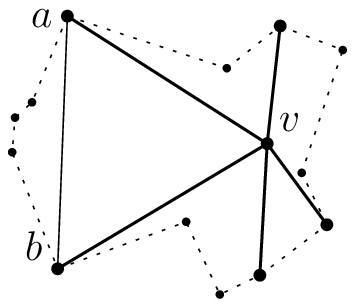}\label{fig:v_edge}}
\caption{ Vertex $v$ and its neighbors. (a) Vertices $a$ and $b$ do not have
direct visibility and the triangle \mpoly{a,b,v} is not empty. (b) A vertex
$u$ is added suitably close to $v$ and connected to $v$, $a$, and $b$. (c) The
output of \conv on the quadrilateral $\langle a,b,v,u \rangle$. (d) Vertex $u$
and its incident edges can be removed in order to insert edge $(a,b)$.
}\label{fig:conv_face}
\end{figure}

Augment $\Sigma_u$ to the drawing $\Theta$ of a maximal plane graph $T$
by first adding three vertices $p$, $q$, and $r$ to $\Sigma_u$, so that triangle
\mpoly{p,q,r} completely encloses the rest of the drawing, and
then adding dummy edges~\cite{chazelle-tsplt-91} till a maximal plane graph is
obtained. If edge $(a,b)$ has been added in this augmentation (this can happen
if $a$ and $b$ share a face not having $v$ on its boundary), subdivide $(a,b)$
in $\Theta$ (namely, replace edge $(a,b)$ with edges $(a,w)$ and $(w,b)$,
placing $w$ along the straight-line segment connecting $a$ and $b$) and
triangulate the two faces vertex $w$ is incident to.

Next, apply the algorithm described in~\cite{aacdfl-mpgdpns-13-c}, that we call \conv, to construct a morphing of $\Theta$ into a drawing $\Theta'$ of $T$ in which polygon \mpoly{a,v,b,u} is convex. The input of algorithm \conv consists of a planar straight-line drawing $\Gamma^*$ of a plane graph $G^*$ and of a set of at most five vertices of $G^*$ inducing a biconnected outerplane graph not containing any other vertex
in its interior in $\Gamma^*$. The output of algorithm \conv is a sequence of
$O(n)$ linear morphing steps transforming $\Gamma^*$ into a drawing of $G^*$ in
which the at most five input vertices bound a convex polygon.
Since, by construction, vertices $a,v,b,u$ satisfy all such requirements, we
can apply algorithm \conv to $\Theta$ and to $a,v,b,u$, hence obtaining a morphing with $O(n)$ steps transforming $\Theta$ into the desired drawing $\Theta'$ (see \Cref{fig:v_convex}).

Let $\Sigma_u'$ be the drawing of $S_u$ obtained by restricting $\Theta'$ to vertices and edges of $S_u$. Since \mpoly{a,v,b,u} is a convex polygon containing no vertex of $S_u$ in its interior, edge $(u,v)$ can be removed from $\Sigma_u'$ and an edge $(a,b)$ can be introduced in $\Sigma_u'$, so that the resulting drawing $\Sigma'$ is planar and cycle $(a,b,v)$ does not contain any vertex in its interior (see \Cref{fig:v_edge}).

Once edge $(a,b)$ has been added to $S$ (either in $\Sigma$ or after the
described procedure transforming $\Sigma$ into $\Sigma'$), if $\deg(v)=2$ then
$v$ is both $a$-contractible and $b$-contractible. Otherwise, consider a new
pair of consecutive vertices of $v$ not creating an empty triangular face with
$v$, if any, and apply the same operations described before.

Once every pair of consecutive vertices has been handled, vertex $v$ is contractible in $S$. Let $\Sigma_v$ be the current drawing of
$S$. Augment $\Sigma_v$ to the drawing $\Theta_v$ of a triangulation $T_v$ (by adding three vertices and a set of edges), contract $v$ onto a neighbor $w$ such
that $v$ is $w$-contractible (one of such neighbors always exists, given that $v$ is contractible), and apply \conv to the resulting drawing $\Theta'_v$ and to the neighbors of $v$ to construct a morphing $\Theta'_v$ to a drawing $\Sigma_v'$ in which the polygon defined by such vertices is convex. Drawing $\Gamma'$ of $G$ in which $v$ is $x$-contractible for any neighbor $x$ of $v$ is obtained by restricting $\Sigma_v'$ to the vertices and the edges of $G$. We can now contract $v$ onto $x$ in $\Gamma'$ and recur on the
obtained graph (with $n-1$ vertices) and drawing.

It remains to observe that, given a quasi-contractible vertex $v$, the procedure to construct a pseudo-morph of $\Gamma$ into $\Gamma'$ consists of at most $\deg(v)+1$ executions of \conv, each requiring a linear number of steps~\cite{aacdfl-mpgdpns-13-c}. As $\deg(v) \leq 5$, the procedure to pseudo-morph $\Gamma$ into $\Gamma'$ has $O(n)$ steps. This concludes the proof of \Cref{th:pseudo-general_case}.

\section{Transforming a Pseudo-Morph into a Morph}\label{se:geometry}
\newcommand{\host}{x\xspace} 
\newcommand{\lhost}{l\xspace} 
\newcommand{\rhost}{r\xspace} 
\newcommand{\mhost}{$\host$\xspace}
\newcommand{\mlhost}{$\lhost$\xspace}
\newcommand{\mrhost}{$\rhost$\xspace}
\newcommand{\bm}[1]{\boldmath{$#1$}\xspace}
In this section we show how to obtain an actual morph $M$ from a
given pseudo-morph ${\cal M}$, by describing how to compute
the placement and the motion of any vertex $v$ that has been contracted
during ${\cal M}$. By applying this procedure to
\Cref{th:sp-pseudo-morphing,th:pseudo-general_case}, we obtain a proof of
\Cref{th:sp-morphing,th:general_case}.

Let $\Gamma$ be a drawing of a graph $G$ and let ${\cal
M}=$\mmorph{\Gamma,\dots,\Gamma^*} be a pseudo-morph that consists of
the contraction of a vertex $v$ of $G$ onto one of its neighbors \mhost,
followed by a pseudo-morph ${\cal M}'$ of the graph $G' = G / (v,\host)$, and
then of the uncontraction of $v$.

The idea of how to compute $M$ from $\cal M$ is the same
as in~\cite{aacdfl-mpgdpns-13-c}: Namely, morph $M$ is obtained
by
\begin{inparaenum}[$(i)$]
\item recursively converting ${\cal M}'$ into a morph $M'$;
\item modifying $M'$ to a morph $M'_v$ obtained by adding vertex $v$ (and its
incident edges) to each drawing of $M'$, in a suitable position;
\item replacing the contraction of $v$ onto \mhost, performed in $\cal M$, with
a
linear morph that moves $v$ from its initial position in $\Gamma$ to its
position in the first drawing of $M'_v$; and
\item replacing the uncontraction of $v$, performed in $\cal M$, with a linear
morph that moves $v$ from its position in the last drawing of $M'_v$ to its
final position in $\Gamma^*$.
\end{inparaenum}
Note that, in order to guarantee the planarity of $M$ when adding $v$ to any
drawing of $M'$ in order to obtain $M'_v$, vertex $v$ must lie inside its
kernel. Since vertex \mhost lies in the kernel of $v$ (as \mhost is adjacent to
all the neighbors of $v$ in $G'$), we achieve this property by placing
$v$ suitably close to \mhost, as follows.

At any time instant $t$ during $M'$, there exists an
$\epsilon_t > 0$ such that the disk $D$ centered at \mhost with radius
$\epsilon_t$ does not contain any vertex other than \mhost. Let $\epsilon$ be
the minimum among the $\epsilon_t$ during $M'$. We place vertex $v$ at
a suitable point of a sector $S$ of $D$ according to the following cases.

\begin{description}
\item [Case \subref{fig:v_deg1}: \bm{v} has degree \bm{1} in \bm{G}.] Sector
$S$ is defined as the intersection of $D$ with the face containing $v$ in $G$.
See \Cref{fig:v_deg1}.

\item [Case \subref{fig:deg2_noedge}: \bm{v} has degree \bm{2} in \bm{G}.]
Sector $S$ is defined as the intersection of $D$ with the face containing $v$ in
$G$ and with the halfplane defined by the straight-line passing through \mhost
and \mrhost, and containing $v$ in $\Gamma$. See \Cref{fig:deg2_noedge}.
\end{description}
Otherwise, $\deg(v) \geq 3$ in $G'$. Let $(\rhost,v)$ and $(\lhost,v)$ be the
two edges such that
$(\rhost,v)$, $(\host,v)$, and $(\lhost,v)$ are clockwise consecutive around $v$
in $G$.
Observe that edges $(\rhost,\host)$ and $(\lhost,\host)$ exist in $G'$. Assume
that \mhost, \mrhost, and
\mlhost are not collinear in any drawing of $M'$, as otherwise we can slightly
perturb such a drawing without
compromising the planarity of $M'$. Let $\alpha_i$ be the angle $\widehat{\lhost
\host \rhost}$ in any
intermediate drawing of $M'$.\\
\begin{description}
\item [Case \subref{fig:deg3_convex}: \bm{\alpha_i < \pi}.] Sector $S$ is
defined as the intersection of $D$ with the wedge delimited by edges
$(\host,\rhost)$ and $(\host,\lhost)$. See \Cref{fig:deg3_convex}.
\item [Case \subref{fig:deg3_reflex}: \bm{\alpha_i > \pi}.] Sector $S$ is
defined as the intersection of $D$ with the
wedge delimited by the elongations of $(\host,\rhost)$ and
$(\host,\lhost)$ emanating from \mhost. See \Cref{fig:deg3_reflex}.
\end{description}
\begin{figure}[htb]\label{fig:deg2}
\centering
\subfigure[]{\includegraphics[scale=1]{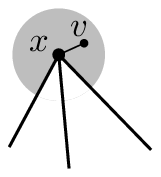}\label{fig:v_deg1}} \hfill
\subfigure[]{\includegraphics[scale=1]{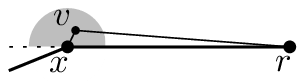}\label{fig:deg2_noedge}}
\hfill
\subfigure[]{\includegraphics[scale=1]{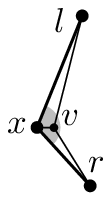}\label{fig:deg3_convex}}\hfill
\subfigure[]{\includegraphics[scale=1]{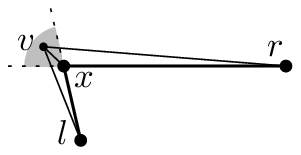}\label{fig:deg3_reflex}}
\hspace{1.1em}
\hspace{1.1em}

\caption{Sector $S$ (in grey) when: \subref{fig:v_deg1} $\deg(v)=1$,
\subref{fig:deg2_noedge} $\deg(v)=2$, and
\subref{fig:deg3_convex}-\subref{fig:deg3_reflex} $\deg(v)\ge 3$.}
\end{figure}

By exploiting the techniques shown in~\cite{aacdfl-mpgdpns-13-c}, the motion of
$v$ can be computed according to the evolution of $S$ over $M'$, thus obtaining
a planar morph $M_v'$. For convenience, we report hereunder the statement of
Lemma~5.2 of~\cite{aacdfl-mpgdpns-13-c}, showing that a contracted vertex can be
placed and moved according to the evolution of a sector defined on one of its
neighbors lying in the kernel.

\begin{lemma}\label{le:genmorph}{\bf (\cite{aacdfl-mpgdpns-13-c})}
Let $\Gamma_1,\dots,\Gamma_k$ be straight-line planar drawings of a $5$-gon $C$
on vertices $a, b, c,
d, e$ in clockwise order such that the morph \mmorph{\Gamma_1,\dots, \Gamma_k}
is planar and vertex $a$ is inside the
kernel of the polygon $C$ at all times during the morph. Then we can augment
each drawing $\Gamma_i$ to a drawing
$\Gamma^p_i$ by adding vertex $p$ at some point $p_i$ inside the kernel of the
polygon $C$ in $\Gamma_i$, and adding
straight line edges from $p$ to each of $a,b,c,d,e$ in such a way that the morph
\mmorph{\Gamma_1,\dots, \Gamma_k} is
planar.
\end{lemma}

Observe that, in the algorithm described in \Cref{se:general_case}, the vertex
$x$ onto which $v$ has been contracted might be not adjacent to $v$ in $G$.
However, since a contraction has been performed, $x$ is adjacent to $v$ in one
of the graphs obtained when augmenting $G$ during the algorithm. Hence, a morph
of $G$ can be obtained by applying the above procedure to the pseudo-morph
computed on this augmented graph and by restricting it to the vertices and edges
of $G$.

\section{Conclusions and Open Problems}\label{se:conclusions}

In this paper we studied the problem of designing efficient algorithms for
morphing two planar straight-line drawings of the same graph. We proved that any
two planar straight-line drawings of a series-parallel graph can be morphed with
$O(n)$ linear morphing steps, and that a planar morph with $O(n^2)$ linear
morphing steps exists between any two planar straight-line drawings of any
planar graph.

It is a natural open question whether the bounds we presented are optimal or
not. We suspect that planar straight-line drawings exist requiring a linear
number of steps to be morphed one into the other. However, no super-constant
lower bound for the number of morphing steps required to morph planar
straight-line drawings is known.

It would be interesting to understand whether our techniques can be extended to
compute morphs between any two drawings of a \emph{partial planar $3$-tree} with
a linear number of steps. We recall that, as observed
in~\cite{aacdfl-mpgdpns-13-c}, a linear number of morphing steps suffices to
morph any two drawings of a \emph{maximal planar $3$-tree}.

\bibliographystyle{splncs03}
\bibliography{bibliography}


\end{document}